\documentclass[reprint, superscriptaddress, amsmath, amssymb, aps, prx]{revtex4-2}

\usepackage{graphicx}
\usepackage{dcolumn}
\usepackage[hang,small,bf]{caption}
\usepackage[subrefformat=parens]{subcaption}
\usepackage{bm}
\usepackage[colorlinks,linkcolor=blue,urlcolor=blue,citecolor=blue]{hyperref}

\usepackage{braket}
\usepackage{comment}
\usepackage{amsthm}

\newtheorem{theorem}{Theorem}
\newtheorem{definition}{Definition}
\newtheorem{lemma}{Lemma}

\newtheorem{assumption}{Assumption}
\newtheorem{problem}{Problem}

\newcommand{\red}{\color{black}} %

\begin{document}


\title{A Polynomial Time Quantum Algorithm for Exponentially Large Scale Nonlinear Differential Equations via Hamiltonian Simulation}%

\author{Yu Tanaka}
\email{Yu.Tanaka@sony.com}
\affiliation{
Advanced Research Laboratory, 
Research Platform,
Sony Group Corporation, 1-7-1 Konan, Minato-ku, Tokyo, 108-0075, Japan
}

\author{Keisuke Fujii}
\affiliation{
Graduate School of Engineering Science, Osaka University,
1-3 Machikaneyama, Toyonaka, Osaka 560-8531, Japan
}

\date{\today}

\begin{abstract}
Quantum computers have the potential to efficiently solve a system of nonlinear ordinary differential equations (ODEs), which play a crucial role in various industries and scientific fields. 
However, it remains unclear which system of nonlinear ODEs, and under what assumptions, can achieve exponential speedup using quantum computers.
In this work, we introduce a class of systems of nonlinear ODEs that can be efficiently solved on quantum computers,
where the efficiency is defined as solving the system with computational complexity of $O(T {\rm log}(N) {\rm polylog}(1/\epsilon))$, where $T$ is the evolution time, $\epsilon$ is the allowed error, and $N$ is the number of variables in the system. 
Specifically, we employ the Koopman-von Neumann linearization to map the system of nonlinear ODEs to Hamiltonian dynamics and find conditions where the norm of the mapped Hamiltonian is preserved and the Hamiltonian is sparse.
This allows us to use the optimal Hamiltonian simulation technique for solving the nonlinear ODEs with $O({\rm log}(N))$ overhead.
Furthermore, we show that the nonlinear ODEs include a wide range of systems of nonlinear ODEs, such as the nonlinear harmonic oscillators and the short-range Kuramoto model. 
This is the first concrete example of solving systems of nonlinear ODEs with exponential quantum speedup by the Koopman-von Neumann linearization, although it is noted that this assumes efficient preparation of the initial state and computation of the output.
These findings contribute significantly to the application of quantum computers in solving nonlinear problems. 
\end{abstract}

\maketitle


\section{Introduction} \label{Sec:1}

Nonlinear differential equations play a very important role in industry as well as in scientific fields such as physics and mathematics. 
If quantum computers can be used to analyze such nonlinear differential equations and accelerate the computation, the applications of quantum computers will expand greatly. 
Besides, the relationship between nonlinear differential equations and machine learning is also attracting attention recently~\cite{chen2018neural}. This could be a new route for the application of quantum computers to the field of machine learning. 
However, since quantum systems are essentially linear, and hence nontrivial treatment is required to apply quantum computers for solving nonlinear differential equations.

Attempts to linearize nonlinear differential equations have long been studied. 
The well-known ones are the Carleman linearization~\cite{carleman1932application} and the Koopman-von Neumann linearization~\cite{koopman1931hamiltonian}. 
These two linearizations are also closely related to quantum mechanics, and Ref.~\cite{Kowalski1997-ju} discusses how to embed and linearize nonlinear differential equations into quantum mechanical systems in a unified framework.
In this context, the Carleman linearization maps 
variables of the system including their higher order nonlinear terms 
directly to the complex probability amplitudes of a quantum system. 

Quantum algorithms for solving nonlinear differential equations using this mapping have been proposed~\cite{Liu2021-wq}. 
In Ref.~\cite{Liu2021-wq}, the authors develop the quantum Carleman linearization, 
a quantum algorithm for dissipative quadratic ordinary differential equations.
That algorithm has complexity $O(T^2 q {\rm poly}(\log T, \log N, \log 1/\epsilon)/\epsilon)$, 
where $T$ is the evolution time, $\epsilon$ is the allowed error, and $q$ measures decay of the solution.
In the quantum Carleman linearization, two infinite dimensional linearizations, the Carleman linearization and the forward Euler method, are applied, and the linearized problem can be solved by the high-precision quantum linear system algorithm (QLSA)~\cite{harrow2009quantum}.
In applying QLSA, the sparsity and the condition number of the matrix are important parameters. The quantum Carleman linearization only inherits the number of terms in the equation as the sparsity.
And, two cutoff parameters introduced for those linearizations as accuracy parameters give the bound of the condition number.
Further, the complexity depends on the decay $q$.
This is inevitable because the linearized differential equations are not Hamiltonian dynamics by applying the Carleman linearization.

On the other hand, the Koopman-von Neumann linearization embeds classical variables via a position state in a continuous variable quantum system.
For its discretization, we can employ a number state and its corresponding orthogonal polynomials.
The advantage of the latter approach is that such an embedding maps the nonlinear differential equations to Hamiltonian dynamics. 
While an infinite-dimensional Hilbert space must be truncated by regulating the total particle number, Ref.~\cite{Engel2021-vl} showed that the truncation error is $O(\eta^m)$, where $\eta$ is  a dimensionless parameter characterizing the strength of the nonlinearity and $m$ is the total particle number. 
Furthermore, the Hamiltonian simulation~\cite{Low2019hamiltonian} was shown to be applied by giving a block encoding, a technique of encoding a matrix as a block of a unitary~\cite{10.1145/3313276.3316366}, of the Hamiltonian into which the nonlinear system was embedded.
Then, the time complexity was shown to be characterized by the max norm of the matrix embedded in the block encoding, which finally depends on the nonlinear equations.

However, the computational complexity is not fully understood as follows. 
One problem is that it has been unclear whether or not the mapped Hamiltonian is sparse and whether or not the norm of the Hamiltonian dynamics is conserved.
Another and more serious problem is that the nonlinearity $\eta$ was assumed to be sufficiently small, which implicitly depends on the max norm of the mapped Hamiltonian~\footnote{Because of this, the previous version of this paper imposed a condition $\eta T \sim 1$.}.
Therefore it is still unclear for what kind of nonlinear differential equations can be solved using a quantum computer with what computational complexity
in terms of the problem parameters, system size $N$, accuracy $\epsilon$, and simulation time $T$.
This motivates to evaluate the computational complexity of solving nonlinear differential equation based on the Koopman-von Neumann linearization more in detail.

To address these difficulties, we present a class of systems of nonlinear ordinary differential equations (ODEs) that can be reducible to the efficient Hamiltonian simulation. (See Fig.~\ref{fig:example_nlde}.)
Here, solving a system of ODEs means deriving a quantity described by a multivariate polynomial of finite degree of the variables in the system that have evolved for a finite time.
The efficiency means that the time evolution of the system of ODEs can be simulated with a computational complexity of $T {\rm log}(N) {\rm polylog}(1/\epsilon)$, where $T$ is the evolution time, $\epsilon$ is the allowed error, and $N$ is the number of variables in the proposed ODEs.
\begin{figure}[t]
    \centering
    \includegraphics[keepaspectratio, scale=0.15]{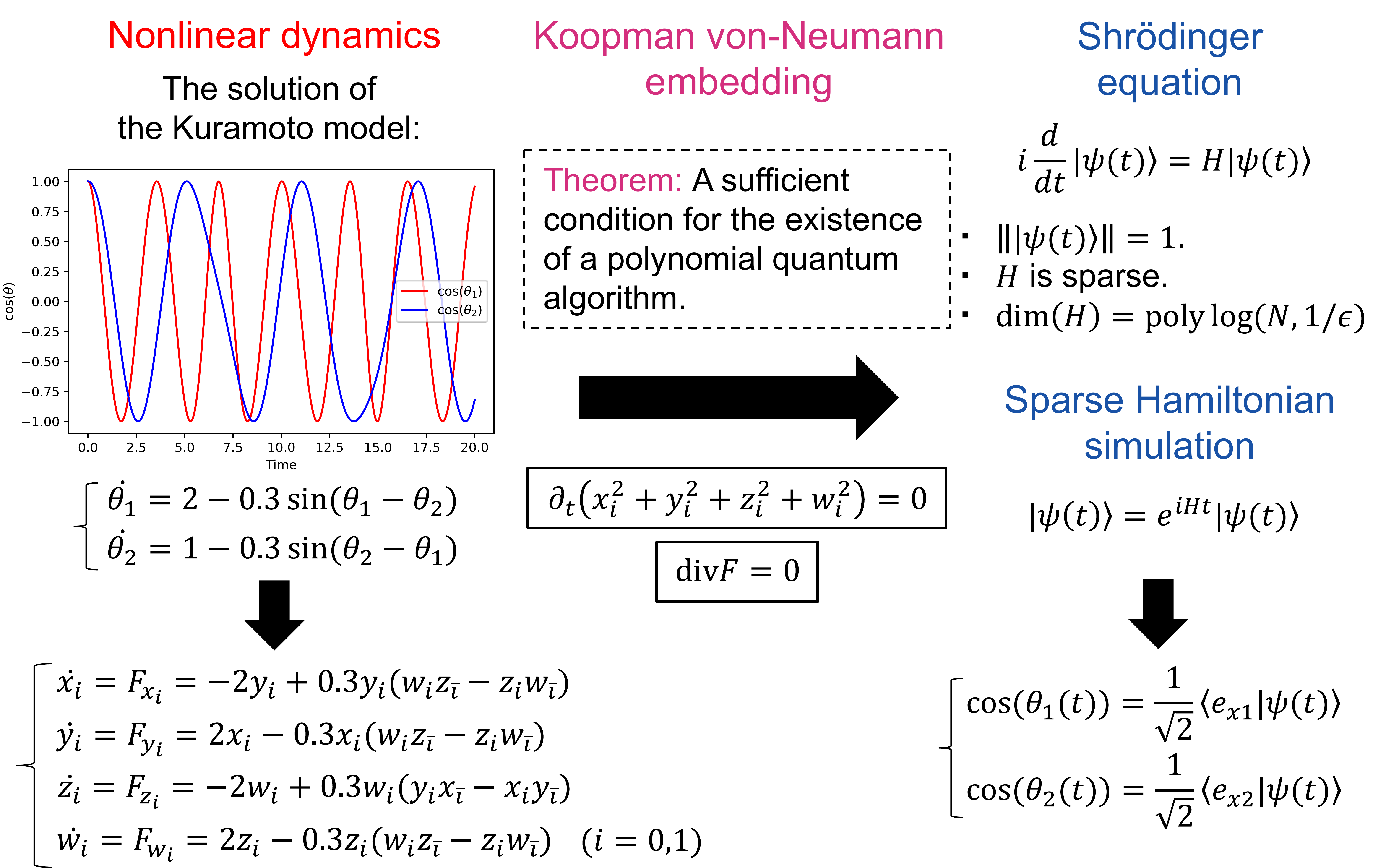}
    \caption{\raggedright An example of embedding the Kuramoto model into Hamiltonian dynamics. First, the Kuramoto model is rewritten into nonlinear differential equations satisfying two conditions. One is ${\rm div}F = 0$, which is required to prevent the norm decay of the quantum state into which the solution $(x_t, y_t, z_t)$ embedded. The other is that $\partial_t(x_i^2+y_i^2+z_i^2+w_i^2) = 0$, which requires the dynamical system to be a conservative system. (Note that some non-conservative systems can be rewritten as the nonlinear ODEs reducible to the efficient Hamiltonian simulation by introducing dummy variables in Sec.~\ref{Sec:6}.)} \label{fig:example_nlde}
\end{figure}

There are two major challenges to solving this problem: first, the norm is not conserved when mapped to Hamiltonian dynamics. 
The second and the most important one is that the mapped Hamiltonian is not sparse, making it impossible to efficiently apply the Hamiltonian simulation. 
We show that the proposed ODEs can resolve these two conditions, and they can be solved exponentially efficiently for the number of $N$ variables.

Furthermore, we rigorously evaluate the truncation error in a similar way to Ref.~\cite{Mizuta2023optimalhamiltonian}, which guarantees that the number of qubits required to achieve an accuracy $\epsilon$ scales $\mathrm{poly} \log (1/\epsilon)$.
This compliments the argument given in Ref.~\cite{Engel2021-vl} and allows us to argue the computational complexity more rigorously.

The proposed ODEs are embedded into Hamiltonian dynamics and applied to the Hamiltonian simulation~\cite{10.1145/3313276.3316366}, thus they depends on neither condition number nor the decay of the solution. 
While the reducible ODEs are a restricted class of systems of ODEs, we claim the broad applications by showing that the proposed ODEs include a wide range of systems of nonlinear ODEs, e.g., classical linear and nonlinear harmonic oscillators, and the short-range Kuramoto model~\cite{RevModPhys.77.137}. 
The short-range Kuramoto model considers locally coupled oscillators, where coupling strength decays with distance, while the classical Kuramoto model describes phase synchronization in globally coupled ones. 

{\red
We identified the mathematical conditions that allow nonlinear ODEs to be tackled on a quantum computer via the Koopman-von Neumann linearization method, without relying on the Carleman linearization. 
Based on these conditions, we were able to find the two specific examples mentioned above, although they are toy models. 
It is expected that this method can be extended to a wider range of practical nonlinear differential equations in the future.
Furthermore, while we focused on Hermite polynomials, note that the same approach applies to other orthogonal polynomials.
}

The rest of this paper is organized as follows.
In Sec.~\ref{Sec:2}, we briefly review an existing linearization technique of nonlinear ODEs into a quantum system using the position operator embedding.
In Sec.~\ref{Sec:3}, 
we introduce a class of nonlinear ODEs and show that the proposed ODEs are embedded into Hamiltonian dynamics. 
We evaluate the upper bounds of the max norm and the spectral norm of the Hamiltonian associated with the proposed ODEs.
Then, we prove the approximation theorem of the cut-offed Hamiltonian dynamics.
In Sec.~\ref{Sec:4}, we introduce an approximate representation to reduce the spatial complexity.
In Sec.~\ref{Sec:5}, we show the algorithm of simulating the proposed ODEs by assuming the sparse access model and the quantum singular value transformation~\cite{10.1145/3313276.3316366}.
In Sec.~\ref{Sec:6}, as some applications, we show that the proposed ODEs include classical linear and nonlinear harmonic oscillators, and the short-range Kuramoto model~\cite{RevModPhys.77.137}.
Section~\ref{Sec:7} is devoted to a conclusion.

\section{Preliminary: Reduction to the evolution equation in Hilbert space}
\label{Sec:2}
In this section, we follow Ref.~\cite{Kowalski1997-ju} and introduce a method for embedding a nonlinear system into an infinite dimensional linear system.
For simplicity, let us consider one-dimensional case
\begin{align}
    \frac{dx}{dt} = F(x), \label{1d_NL_system}
\end{align}
where $F(x)$ is an analytic function.
To map Eq.~(\ref{1d_NL_system}) to a linear system, first, we introduce a hermitian operator $\hat{x}$ with the complete set of eigenvectors such that
\begin{eqnarray}
\hat{x} \ket{x} = x \ket{x}, \label{eq1}
\end{eqnarray}
and 
\begin{eqnarray}
\braket{n|x} = w(x)^{1/2} p_n(x), n \in \mathbb{Z}_{\ge 0} \label{eq2}
\end{eqnarray}
where $\{ \ket{n} \ | \ n \in \mathbb{Z}_{\ge 0} \}$ is the basis set in the occupation number representation and $p_n(x)$ is a normalized classical orthogonal polynomial defined by the trinomial recurrence formula
\begin{eqnarray}
p_{n+1}(x) = (A_n + B_n x)p_n(x) - C_n p_{n-1}, n \in \mathbb{Z}_{\ge 0},
\end{eqnarray}
where $p_{-1}(x) = 0$.
Using creation and annihilation operators $\hat{a}^{\dagger}$ and $\hat{a}$, \textit{i.e.},
\begin{eqnarray}
\hat{a}^{\dagger} \ket{n} = \sqrt{n+1} \ket{n+1}, \hat{a} \ket{n} = \sqrt{n} \ket{n-1},
\end{eqnarray}
we have the representation of $\hat{x} = x(\hat{a}^\dagger, \hat{a})$ as the operator function of $\hat{a}^\dagger$ and $\hat{a}$. 
Second, introducing a hermitian operator $\hat{k} = k(\hat{a}^{\dagger},\hat{a})$ satisfying the commutation relation $[\hat{x}, \hat{k}] = i X(\hat{x})$, where $X$ is determined by the choice of the classical orthogonal polynomials,
the following vector 
\begin{eqnarray}
\ket{\tilde{x}(t)} = \exp \Bigl[ \frac{1}{2} \int_0^t F'(x) d\tau \Bigr] \ket{x(t)} \label{vec}
\end{eqnarray}
satisfies the Shr\"{o}dinger equation
\begin{eqnarray}
i \frac{d}{dt} \ket{\tilde{x}(t)} = \hat{H} \ket{\tilde{x}(t)}, \label{shr}
\end{eqnarray}
where the Hamiltonian $\hat{H}$ is written as
\begin{eqnarray}
\hat{H} = \frac{1}{2} \Bigl( \hat{k} \frac{F(\hat{x})}{X(\hat{x})} + \frac{F(\hat{x})}{X(\hat{x})} \hat{k} \Bigr) . \label{H}
\end{eqnarray}
Specifically, if we chose Hermitian polynomial, 
$\hat{a}^{\dag},\hat{a}$ are the creation and annihilation operators of the harmonic oscillator, and $\hat{k}$ is the momentum operator, which leads $X=1$.

To generalize the above result for a multidimensional case, we focus on an initial value problem described by the following $N$-dimensional ODE
\begin{eqnarray}
\frac{dx_i}{dt} &=& F_i(\mathbf{x}),\ i \in [N], \label{NL_system} 
\end{eqnarray}
where $F_i(\mathbf{x})$ is a multivariate analytic function and $x_i (t)$ is a function of $t$ on the interval $[0, T]$.
And let us introduce hermitian operators $\hat{x}_i$ 
with the complete set of eigenvectors $\ket{x_i}$.
Then, for the nonlinear dynamical system,
defining the Hamiltonian $\hat{H}$ as
\begin{eqnarray}
\hat{H} = \frac{1}{2} \sum_{i=1}^N \Bigl( \hat{k}_i \frac{F_i(\hat{\textbf{x}})}{X(\hat{x}_i)} + \frac{F_i(\hat{\textbf{x}})}{X(\hat{x}_i)} \hat{k}_i \Bigr), \label{eq:MH}
\end{eqnarray}
the Shr\"{o}dinger equation (\ref{shr}) holds for the following tensor product of states
\begin{eqnarray}
\ket{\tilde{\textbf{x}}(t)} 
&=&
\exp \Bigl[ \frac{1}{2} \int_0^t \mathrm{div} F(\mathbf{x}) d\tau \Bigr]
\ket{\textbf{x}(t)}, \label{sol}
\end{eqnarray}
where $\ket{\textbf{x}(t)} := \bigotimes_{i=1}^N \ket{x_i(t)}$.

{\red
The state $\ket{\tilde{\textbf{x}}(t)}$ in Eq.~(\ref{sol}) does not preserve the norm in general. 
Thus, to avoid the global decaying term, 
we assume that 
\begin{eqnarray}
\mathrm{div} F(\mathbf{x}) = \sum_{i=1}^N \frac{\partial F_i(\mathbf{x})}{\partial x_i} = 0. \label{condition1}
\end{eqnarray}
Furthermore, expanding Eq.~(\ref{sol}) in the occupation number representation, we have
\begin{eqnarray}
\ket{\textbf{x}(t)} = \bigotimes_{i=1}^N \left( w(x_i)^{\frac{1}{2}} \sum_{n_i=0}^{\infty} p_{n_i}(x_i) \ket{n_i} \right). \label{Expand}
\end{eqnarray}
The factor $w(\mathbf{x}) := \prod_{i=1}^N w(x_i(t))$ in Eq.~(\ref{Expand}) can complicate the representation of desired outputs. 
Actually, since the output quantity can be obtained as an inner product
\begin{eqnarray}
    \braket{c \mid \mathbf{x}(t)} = w(\mathbf{x}(t))^{\frac{1}{2}} \sum_{\mathbf{n}: |\mathbf{n}| \le b} c_{\mathbf{n}} \prod_{i=0}^{N} p_{n_i}(x_i), \label{eq:output}
\end{eqnarray}
where $\ket{c} = \sum_{\mathbf{n}: |\mathbf{n}| \le b} c_{\mathbf{n}} \ket{\mathbf{n}}$ is an arbitrary state in the Hilbert space restricted to the subspace bounded by the total occupation number $b$, we need to remove $w(\textbf{x}(t))$ from the inner product. In general, the value $w(\textbf{x}(t))$ after time evolution is unknown and its measurement cost is also non trivial.
Thus, we assume that $w(\textbf{x}(0))$ is known and we require that
\begin{eqnarray}
\dot{w}(\mathbf{x}) = \sum_{i=1}^N \frac{\partial w(\mathbf{x})}{\partial x_i} F_i(\mathbf{x}) = 0. \label{condition2}
\end{eqnarray}

From the above considerations, we present our problem setting for the system of nonlinear differential equations.
\begin{problem} \label{problem_1}
    We consider a system of $N$-dimensional ODEs
    \begin{eqnarray}
        \frac{d x_i}{dt} = F_i (\mathbf{x}), \in [N],
    \end{eqnarray}
    where $x_i(t)$ is a function of $t$ on the interval $[0, T]$ and $F_i (\mathbf{x})$ is a multivariate analytic function.
    The system satisfies that ${\rm div} F(\mathbf{x}) = 0$ and $\dot{w}(\mathbf{x}) = 0$, where $w$ is the weight function of a classical orthogonal polynomial system. 
    Then, for a given initial value $\mathbf{x}(0)$, the problem is to approximate the output $p(\mathbf{x}(T))$ of a $b$-degree multivariate polynomial within the accuracy $\epsilon$.  
\end{problem}
}

{\red
On a classical computer, solving Problem~\ref{problem_1} requires ${\rm poly}(N)$ overhead in space and time. 
Our goal here is to solve such a problem in ${\rm polylog}(N)$ qubits and computation time on a quantum computer. 

We see that Problem~\ref{problem_1} can be divided into three parts.
The first part is the preparation of the initial state $\ket{\mathbf{x}(0)}$, which encodes the initial value $\mathbf{x}(0)$.
The second part is the efficient Hamiltonian simulation using the Hamiltonian given by Eq.~(\ref{eq:MH}). 
The third part is the evaluation of the output as the inner product $\braket{c \mid \mathbf{x}(T)}$ in Eq.~(\ref{eq:output}).

In this paper, we focus on the second problem and make the assumptions regarding the first and third ones.
The assumption for the initial state preparation is given as 
\begin{assumption} \label{assumption_1}
    Given an initial value $\mathbf{x}(0)$, we assume an oracle that prepares the initial state $\ket{\mathbf{x}(0)}$ in Problem~\ref{problem_1}. 
\end{assumption}
The reason for making Assumption~\ref{assumption_1} is that the distribution of initial values differs for each problem, and whether it is possible to efficiently encode these initial values into a quantum computer depends on whether the data structure of these initial values can be used.

On the other hand, the assumption for the evaluation of the output is given as
\begin{assumption} \label{assumption_2}
    We assume an oracle that prepares the quantum state $\ket{c}$, such that the inner product $\braket{c \mid \mathbf{x}(t)}$ represents the output multivariate polynomial in Problem~\ref{problem_1}.
\end{assumption}

To argue the efficient Hamiltonian simulation using the Hamiltonian in Eq.~(\ref{eq:MH}), we need to evaluate the Hamiltonian $\hat{H}_m$ on the Hilbert space $\mathcal{H}_m$ restricted to the subspace bounded by the total occupation number $m$, which is also the accuracy parameter in the algorithmic view point. 
Therefore, we focus on deriving a sufficiently large scale of the total occupation number $m$ to achieve the required output accuracy $\epsilon$ through the analysis of the simulation error defined in the following.

\begin{definition} \label{def:error}
    For a given Hamiltonian $\hat{H}$ on a Hilbert space $\mathcal{H}$ and an initial state $\ket{\psi} \in \mathcal{H}$, we denote $\hat{H}_m$ and $\ket{\psi_m}$ as the Hamiltonian on the subspace $\mathcal{H}_m$ bounded by the total occupation number $m$ and the quantum state in $\mathcal{H}_m$, respectively. Then, for any quantum state $\ket{c} \in \mathcal{H}_b$, we define the simulation error by
    \begin{eqnarray}
        \left| \bra{c} e^{-i \hat{H}t} \ket{\psi} - \bra{c} e^{-i \hat{H}_m t} \ket{\psi_m} \right|. \label{eq:simulation_error}
    \end{eqnarray}
\end{definition}

In the next section, we introduce a restricted class of ODEs that can be solved using the Koopman-von Neumann embedding. 
By evaluating the error Eq.~(\ref{eq:simulation_error}) defined in Eq.~(\ref{eq:simulation_error}), we derive the total occupation number $m$ required to achieve the desired accuracy $\epsilon$.

In the rest of the paper, for a given Hilbert space $\mathcal{H}$, a Hamiltonian $\hat{H}$ and a quantum state $\ket{\psi}$, we denote $\mathcal{H}_m$, $\hat{H}_m$ and $\ket{\psi_m}$ as the Hilbert space restricted to the subspace bounded by the total occupation number $m$, the truncated Hamiltonian on $\mathcal{H}_m$ and the truncated quantum state in $\mathcal{H}_m$, respectively. 
}

\section{Nonlinear ODEs reducible to Hamiltonian dynamics} \label{Sec:3}

We introduce one restricted ODE class solvable with the Koopman-von Neumann embedding, which means that the ODE class is reduced to the evolution equation through Hermite polynomial expansion and satisfies Eq.~(\ref{condition1}) and Eq.~(\ref{condition2}).

\begin{definition}
    \label{def:quantum_solvable_ODE}
    A system of $N$-dimensional ODEs in the interval $[0,T]$ solvable with the Koopman-von Neumann embedding is a system of nonlinear ODEs defined by a family $\mathcal{S}$ of index sets each of which determines the variables engaged in an interaction, and coupling constants $\alpha _{p \to i}$ of the interactions:
    \begin{eqnarray}
        \frac{d x_i}{dt} 
        = F_i(\mathbf{x}) = \sum_{\substack{p \in \mathcal{S} \\ : i \in p}} \alpha _{p \to i} \prod_{j \in p\setminus i} x_j, i \in [N], \label{eq:quantum_solvable_ODE}
    \end{eqnarray}
    where the index sets $p \in \mathcal{S}$, and the coupling constants $\alpha _{p \to i}$ satisfy the followings:
     \begin{itemize}
        \item[1] For any $p \in \mathcal{S}$, $2 \leq |p| \leq d$.
        \item[2] For any $i \in [N]$, there exists at least one and at most $c$ sets $p$ such as $i \in p$.
        \item[3] For any $p = (i_1,\cdots, i_l) \in \mathcal{S}$, real numbers $\alpha_{p \to i_1}, \cdots, \alpha _{p \to i_l} \in \mathbb{R}$ are assigned such that $\sum_{i \in p} \alpha_{p \to i} = 0$. 
    \end{itemize}
\end{definition}
Note that the other ODE class with terms in which the degree of each variable is two or higher can be reduced to Eq.~(\ref{eq:quantum_solvable_ODE}) by introducing dummy variables~\cite{Engel2021-vl} and considering multiple systems, each identical to the original, including in their initial conditions.
In Sec.~\ref{Sec:6}, we show the applicability of the proposed system of ODEs by reducing several systems to it.

The proposed system of ODEs satisfies Eq.~(\ref{condition1}) trivially, since the term of Eq.~(\ref{eq:quantum_solvable_ODE}) does not contain $x_i$.
It also satisfies Eq.~(\ref{condition2}) as follows. 
Considering that the weight function of the Hermite polynomial system is $w(x) = e^{-x^2}$, we can check Eq.~(\ref{condition2}) directly:
\begin{eqnarray}
        \dot{w}(\mathbf{x}) 
        &=&
        \sum_{i=1}^N \frac{\partial w(\mathbf{x})}{\partial x_i} F_i(\mathbf{x})
        \propto
        \sum_{i=1}^N \sum_{\substack{p \in \mathcal{S} \\ : i \in p}} \alpha_{p \to i} \prod_{j \in p} x_j \label{eq:sum_N} \\
        &=&
        \sum_{p \in \mathcal{S}} \left( \sum_{i \in p} \alpha_{p\to i} \right) \prod_{j \in p} x_j = 0, \label{eq:sum_S}
\end{eqnarray}
where Eq.~(\ref{eq:sum_S}) is derived as follows.
Focusing one $p \in \mathcal{S}$ in Eq.~(\ref{eq:sum_N}), 
we can see that the term
\begin{eqnarray}
    \alpha_{p\to i} x_i \prod_{j \in p\setminus i} x_j,  
\end{eqnarray}
appears only once at each $i \in p$ through the sum $\sum_{i=1}^N$. 
Putting these terms together for each $p \in \mathcal{S}$, we obtain Eq.~(\ref{eq:sum_S}).

Clearly, the Hamiltonian corresponding to the proposed ODEs is written by
\begin{eqnarray}
    \hat{H} 
    &=& 
    \sum_{i=1}^N \sum_{\substack{p \in \mathcal{S} \\ : i \in p}} \alpha_{p\to i} \hat{k}_i \prod_{j \in p\setminus i} \hat{x}_j \label{eq:map_H} \\
    &=& 
    \sum_{p \in \mathcal{S}} \left( \sum_{i \in p} \alpha_{p\to i} \hat{k}_i \prod_{j \in p\setminus i} \hat{x}_j \right), \label{eq:q_sol_H}
\end{eqnarray}
where $\hat{x} = (\hat{a} + \hat{a}^{\dagger})/\sqrt{2}$ and $\hat{k} = i (\hat{a}^{\dagger} - \hat{a})/\sqrt{2}$.
{\red

Then, as stated in Def.~\ref{def:error}, we analyze the error given by Eq.~(\ref{eq:simulation_error}) that occurs when the initial state is time evolved in the Hilbert space $\mathcal{H}_m$ restricted to the subspace bounded by the total occupation number $m$. 
First, we prove the following technical lemma.
\begin{lemma} \label{lem:1}
    For the Hamiltonian $\hat{H}$ given by Eq.~(\ref{eq:map_H}), which corresponds to a system of $N$-dimensional ODEs defined by Eq.~(\ref{eq:quantum_solvable_ODE}), we consider the Hamiltonian $\hat{H}_m$ on the Hilbert space restricted to the subspace bounded by the total occupation number $m$. 
    Then, the Hamiltonian $\hat{H}_m$ is $O(m)$-sparse, the max norm of $\hat{H}_m$ is $O(m^{d/2})$, and the spectral norm of $\hat{H}_m$ is $O(m^{d/2})$.
\end{lemma}

\begin{proof}
    See Appendix~\ref{apdx:A} for all the lemmas.
\end{proof}

The reason of exponential dependence of the norm on $d$ is as follows.
When the system of ODEs is embedded in the sum of local Hamiltonians, the variables are replaced by creation and annihilation operators. Consequently, $d$ corresponds to the maximum degree of creation and annihilation operators in the local Hamiltonian. As a result, the norm is $O(m^{d/2})$.

The spectral norm of the Hamiltonian is crucial for analyzing the error in the time evolution of the state. This is because if the norm is too large, the simulation becomes infeasible, as the approximation error can no longer be guaranteed to converge asymptotically.

Additionally, we introduce a rescaled Hamiltonian 
\begin{eqnarray}
    \hat{H} 
    =
    \sum_{i=1}^N\sum_{\substack{p \in \mathcal{S} \\ : i \in p}} \frac{\alpha _{p \to i}}{\Delta^{|p|-2}} \hat{k}_i \prod_{j \in p\setminus i} \hat{x}_j, \label{eq:rescaled_H}
\end{eqnarray}
that incorporates a scaling parameter $\Delta$.
The rescaling is defined by the transformation $x \to \Delta x$.
The rationale behind this rescaling is technical: it is intended to suppress the divergence of inequalities used in error evaluation.
Here, we present our main result in the following.

\begin{theorem} \label{thm:main}
    We consider a system of $N$-dimensional ODEs in the interval $[0, T]$ given by Def.~\ref{def:quantum_solvable_ODE}, the corresponding rescaled Hamiltonian $\hat{H}$ given by Eq.~(\ref{eq:rescaled_H}), and an initial state $\ket{\psi(0)}$. 
    Then, for a given required accuracy $\epsilon$ and a quantum state $\ket{c} \in \mathcal{H}_b$, if $m = O\left( \log \left( 1/\epsilon \right) \right)$ and $\Delta = O\left( {\rm poly} \log \left( 1/\epsilon \right) \right)$ are sufficiently large, the following equation
    \begin{eqnarray}
        \left| \Bra{c} e^{-i\hat{H}t} \Ket{\psi(0)} - \Bra{c} e^{-i\hat{H}_mt} \Ket{\psi_m(0)} \right| < \frac{\epsilon}{\Delta^b}
    \end{eqnarray}
    holds for $t \in \left[0, \| \hat{H}_m \|_{\max} T \right]$. 
\end{theorem}
}

\begin{proof}
    {\red The proof is outlined as follows. 
    To evaluate the simulation error, we need to observe how $\Bra{c} e^{-i\hat{H}_mt} \Ket{\psi_m(0)}$ behaves for sufficiently large $m$.
    This requires to evaluate the transition amplitude in Fock space, which can be analyzed using the same technique as the Lieb-Robinson bound in the Floquet-Hilbert space proposed in Ref.~\cite{Mizuta2023optimalhamiltonian}. 
    In this evaluation, we also use the result of Lemma~\ref{lem:1}. 
    Next, following the proof technique in Ref.~\cite{Engel2021-vl}, we demonstrate our claim by expanding $\Bra{c} e^{-i\hat{H}t} \Ket{\psi(0)}$ and $\Bra{c} e^{-i\hat{H}_mt} \Ket{\psi_m(0)}$ as power series in the nonlinear parameter and comparing their expansions.}
    See Appendix~\ref{apdx:B} in detail.
\end{proof}

{\red
Whether the approximation holds depends on the behavior of the higher-order terms when the time evolution operator, defined by the Hamiltonian $\hat{H}_m$, is expanded as a series of Hermite polynomials.
The total occupation number $m$ ensures a sufficiently large dimension, preventing the loss of information in the output caused by the truncation of the Hilbert space.
Theorem~\ref{thm:main} demonstrates that these guarantees can be achieved for $m$ and $\Delta$, with scalings of $O\left( \log \left( 1/\epsilon \right) \right)$ and $O\left( {\rm poly} \log \left( 1/\epsilon \right) \right)$, respectively.
}

As shown later, since we apply the block encoding of $\hat{H}_m / ||\hat{H}_m||_{max}$ to the optimal block-Hamiltonian simulation~\cite{10.1145/3313276.3316366}, the effective time required to simulate time $T$ is $||\hat{H}_m||_{max} T$ from Lemma~\ref{lem:1}. 
More importantly, Theorem~\ref{thm:main} shows that if $m$ is sufficiently larger than $b$, an approximated position state $\ket{x}$ in the finite $m$ occupation number representation is sufficient to obtain a degree $b$ polynomial output, which motivates the truncation of the Hilbert space.

\section{Truncation of the space and initial states preparation} \label{Sec:4}
The Hamiltonian in the previous section is not efficient in terms of spatial complexity because the input size is proportional to the number of variables $N$. Thus, let us consider truncating the space of the occupation basis such that the total number is limited by some integer $m$ determined by some required precision. 

We introduce the following encoding~\cite{Engel2021-vl}. 
For any given state $\ket{\mathbf{n}} = \bigotimes_{i=1}^N \ket{n_i}$ such that $|\textbf{n}| \le m$, 
\begin{eqnarray}
    \ket{\mathbf{n}} \mapsto \ket{\mathbf{n}_m} :=
    \ket{\underbrace{\overbrace{i_1 \cdots i_1}^{n_{i_1}}\overbrace{i_2 \cdots i_2}^{n_{i_2}} \cdots \overbrace{i_l \cdots i_l}^{n_{i_l}}}_{m}}, \label{state:trun}
\end{eqnarray}
where $0 \le i_1 < \cdots < i_l \le N$, {\it i.e.}, $\mathbf{n}_m$ is all the combinations that take $m$ integers, allowing for duplicates from $N+1$ integers $\{ 0,1,\cdots, N \}$, in ascending order.
Then, the $m$-truncated subspace is defined to be the Hilbert space spanned by the support $\{ \ket{\mathbf{n}_m} \ | \ |\mathbf{n}| \le m \}$. 
The occupation number representation of an initial state
\begin{eqnarray}
    \ket{\psi(0)} = \sum_{\mathbf{n}} \braket{\mathbf{n} | \psi(0)} \ket{\mathbf{n}}
\end{eqnarray} 
is encoded into the $m$-truncated subspace such as
\begin{eqnarray}
    \ket{\psi(0)} \mapsto \ket{\psi_m(0)} = \frac{1}{\sqrt{L}} \sum_{\mathbf{n} : |\mathbf{n}| \le m} \braket{\mathbf{n} | \psi(0)} \ket{\mathbf{n}_m}, \label{eq:init_state}
\end{eqnarray}
where $L$ is the normalization constant.

{\red Furthermore, in an initial value problem, the occupation number representation of an initial position state $\ket{\mathbf{x}}$ is encoded into the $m$-truncated subspace such as
\begin{eqnarray}
    \ket{\mathbf{x}} 
    &=& 
    \sum_{\mathbf{n}} \braket{\mathbf{n} | \mathbf{x}} \ket{\mathbf{n}} \label{init_x} \\
    &=& 
    w(\mathbf{x})^{1/2} 
    \sum_{\mathbf{n}} p_{\mathbf{n}}(\mathbf{x}) \ket{\mathbf{n}} \\
    \mapsto \ket{\mathbf{x}_m} 
    &=&
    \frac{1}{\sqrt{L}} \sum_{\mathbf{n} : |\mathbf{n}| \le m} p_{\mathbf{n}}(\mathbf{x}) \ket{\mathbf{n}_m}, \label{tran_x}
\end{eqnarray}
where $\ket{\mathbf{n}_m}$ is the $m$-truncated basis in Eq.~(\ref{state:trun}), $L$ is the normalized constant~\footnote{$L$ is a constant thanks to the preservation of $w(\mathbf{x})$, {\it i.e.}, Eq.~(\ref{condition2}).}, and
\begin{eqnarray}
    p_{\mathbf{n}}(\mathbf{x}) = p_{n_{i_1}}(x_{i_1}) p_{n_{i_2}}(x_{i_2}) \cdots p_{n_{i_l}}(x_{i_l}).
\end{eqnarray}
Here, $p_k(x)$ is the normalized $k$-th order Hermite polynomial.
Clearly, the needed qubit size for these initial states is $m \log N$.

Whether the initial state such as $\ket{\psi_m(0)}$ and $\ket{\mathbf{x}_m}$ can be efficiently prepared depends on the access models.
Thus, instead of Assumption~\ref{assumption_1}, we assume the following initial state preparation.
\renewcommand{\theassumption}{1'}
\begin{assumption}
    For a given classical description of an initial state $\ket{\psi(0)}$, there exists an oracle that prepares the $m$-truncated initial state $\ket{\psi_m(0)}$ in the basis given by Eq.~(\ref{state:trun}).
\end{assumption}
\renewcommand{\theassumption}{\arabic{assumption}}

We briefly discuss the computational complexity of initial state preparation. In general, whether a superposition state with arbitrary coefficients can be efficiently prepared depends on the structure of the coefficients themselves, as it is constrained by the known trade-off between time and space complexities~\cite{Zhang2022-bb, Sun2023-ct}. 

Therefore, it is generally unreasonable to assume that the initial state can be prepared in logarithmic time. 
For example, if all the initial values $N$ are independent, the state cannot be prepared in $O(\log N)$ time. 
However, we expect this assumption to hold for problems of practical importance in applications.
}

\section{Computational complexity of Hamiltonian simulation} \label{Sec:5}
In this section, we argue the Hamiltonian simulation of the mapped Hamiltonian of a system of $N$-dimensional ODEs in Def.~\ref{def:quantum_solvable_ODE}.
We assume that we can access the mapped Hamiltonian matrix through the sparse-access oracles~\cite{10.1145/3313276.3316366}.

\begin{definition}
    Let $H \in \mathbb{C}^{2^w \times 2^w}$ be a hermitian matrix that is $s$-sparse, and each element of $H$ has absolute value at most $1$. Suppose that we have access to the following sparse-access oracles acting on two $(w+1)$-qubit registers
    \begin{eqnarray}
        O_r &:& \ket{i}\ket{k} \mapsto \ket{i}\ket{r_{ik}},\ i \in [2^w] - 1, k \in [s], \\
        O_c &:& \ket{l}\ket{j} \mapsto \ket{c_{lj}}\ket{j},\ l \in [s], j \in [2^w] - 1,
    \end{eqnarray}
    where $r_{ij}$ is the index for the $j$-th non-zero entry of the $i$-th row of $H$, or if there are less than $i$ non-zero entries, then it is $j + 2^w$, and similarly $c_{ij}$ is the index for the $i$-th non-zero entry of the $j$-th column of $H$, or if there are less than $j$ non-zero entries, then it is $i+2^w$. Additionally assume that we have access to an oracle $O_H$ that returns the entries of $H$ in a binary description
    \begin{eqnarray}
        O_H : \ket{i}\ket{j}\ket{0}^{\otimes b} \mapsto \ket{i}\ket{j}\ket{H_{ij}},
    \end{eqnarray}
    where $H_{ij}$ is a $b$-bit binary description of the $ij$-matrix element of $H$.
\end{definition}

\subsection{Hamiltonian simulation of ODEs solvable with the Koopman-von Neumann embedding}

Following Refs.~\cite{10.1145/3313276.3316366,Low2019hamiltonian}, we show the optimal block-Hamiltonian simulation of a system of $N$-dimensional ODEs in Def.~\ref{def:quantum_solvable_ODE}.

\begin{theorem} \label{BEODE}
    Let us consider a system of $N$-dimensional ODEs in Def.~\ref{def:quantum_solvable_ODE}
    and the rescaled mapped Hamiltonian 
    \begin{eqnarray}
        \hat{H} = \sum_{i=1}^N \sum_{\substack{p \in \mathcal{S} \\ : i \in p}} \frac{\alpha _{p \to i}}{\Delta^{|p|-2}} \hat{k}_i \prod_{j \in p\setminus i} \hat{x}_j,  \label{eq:re_scaled_H}
    \end{eqnarray}
    where $m$ and $\Delta$ are given as defined in Theorem \ref{thm:main}.
    Denote that
    \begin{eqnarray}
        \eta = \max_{i \in p, p \in \mathcal{S}} \left| \frac{\alpha_{p\to i}}{\Delta^{|p|-2}} \right|.
    \end{eqnarray}
    
    Under the assumption of the sparse-access oracles, we can implement a $(c2^dm, \lceil m \log N \rceil + 3, \epsilon)$-block encoding $U$ of $\tilde{H} = \hat{H} / (\eta d \left( m/2 \right)^{d/2})$ with a single use of $O_r, O_c$, two uses of $O_H$ and additionally using $O(m \log N + \log^{2.5}(\frac{c^2 4^d m^2}{\epsilon}))$ one and two qubit gates while using $O(b, \log^{2.5}(\frac{c^2 4^d  m^2}{\epsilon}))$ auxilliary qubits, where $b$ is the number of bits to represent a binary description of $\tilde{H}_{ij}$.
\end{theorem}

\begin{proof}
    The sparsity of the Hamiltonian is at most $c2^dm$ as shown in Lemma~\ref{lem:1}.
    The enough qubit size for the truncated subspace is $w = \lceil m \log N \rceil$.  
    Then, apply Lemma 48 in Ref.~\cite{10.1145/3313276.3316366}. 
\end{proof}

Then, we apply the above block encoding to the optimal block-Hamiltonian simulation~\cite{10.1145/3313276.3316366}.

\begin{theorem} \label{thm:H_sim}
    For a given $(c2^dm, \lceil m \log N \rceil + 3, 0)$-block encoding $U$ of $\tilde{H} = \hat{H} / (\eta d \left( m/2 \right)^{d/2})$ in Theorem \ref{BEODE}, we can implement $(1, \lceil m \log N \rceil + 5, 6\epsilon)$-block encoding $e^{i\hat{H}t}$, with $3f(\alpha t, \epsilon)$ uses of $U$ or its inverse, $3$ uses of controlled-$U$ or its inverse and with $O(m \log N f(\alpha t, \epsilon))$ two-qubit gates and using $O(1)$ auxilliary qubits, where
    \begin{eqnarray}
        f\left(\alpha t, \epsilon \right) &=& \Theta \left( 
            \alpha t + \frac{\log(1/\epsilon)}{\log(e + \log(1/\epsilon)/(\alpha t))} 
        \right), \label{eq:gate_of_eH1} \\
        \alpha &=& \frac{e}{4} \eta cd (2m)^{\frac{d}{2}+1}. \label{eq:gate_of_eH2}
    \end{eqnarray}
\end{theorem}

\begin{proof}
    To implement a block-encoding of $e^{i\hat{H}t}$, we only need to be able to implement $e^{i\tilde{H} \eta d(m/2)^{d/2} t}$. Then, apply the optimal block-Hamiltonian simulation~\cite{Low2019hamiltonian} (or Theorem 58 in Ref.~\cite{10.1145/3313276.3316366}) and Corollary 60 in Ref.~\cite{10.1145/3313276.3316366}.
\end{proof}

From Theorem 1, 2, and 3, we can estimate the leading gate complexity of Hamiltonian simulation in the following. First, the total occupation number $m$ is $O(\log (1/\epsilon))$ from Theorem \ref{thm:main}. Second, the gate complexity of the block-encoding of the Hamiltonian $\tilde{H}$ is $O(m\log N)$ from Theorem~\ref{BEODE}. Third, the gate complexity of the block-encoding of $e^{i\hat{H}}t$ is $O(T \log (N) m^{d/2+2})$ from Eq.~(\ref{eq:gate_of_eH1}) and Eq.~(\ref{eq:gate_of_eH2}) of Theorem~\ref{thm:H_sim}. Therefore, the total gate complexity is roughly estimated to be $O(T \log (N) \mathrm{poly}(\log 1/\epsilon))$.

{\red
In Ref.~\cite{Liu2021-wq}, the quantum approach based on the Carleman linearization has a complexity of $O(T^2 q {\rm poly}(\log T, \log N, \log 1/\epsilon)/\epsilon)$, where $q$ is the decay of the solution. 
In contrast, our method offers three advantages: it scales linearly in $T$ (instead of quadratically), depends only on $\log (1/\epsilon)$ (rather than $1/\epsilon$), and does not rely on the decay factor $q$.

On the other hand, the classical approach based on the $p$-th order Runge-Kutta method, when the dependencies among ODE variables are local, the computational complexity scales as $O(T N (1/\epsilon)^{1/p})$. 
Our method scales only logarithmically with $N$, thereby significantly reducing the computational overhead for large systems. 
This improved complexity is made possible by leveraging the structural properties of the system, such as locality in variable dependencies and the source-free condition. 
However, when comparing classical and quantum methods, note that quantum methods assume that the initial state can be efficiently prepared and the output can be efficiently computed.
}

\section{Application}\label{Sec:6}

In this section, we show three examples of ODEs can be reduced to the proposed ODEs in Def.~\ref{def:quantum_solvable_ODE}. 
The first two examples are the classical harmonic oscillators, and the undamped and unforced Duffing equations. 
Recently, a quantum algorithm has been proposed for solving coupled harmonic oscillators, a classical mechanical system, with exponential quantum speedup~\cite{Babbush2023-ju}. While their study can only handle essentially linear systems, we show that our research is capable of addressing nonlinear systems, enabling exponential quantum speedup with respect to the number of variables $N$ in the system. 

The other is the Kuramoto model, a simplification of a model introduced by Kuramoto~\cite{kuramoto2012chemical} to study huge populations of coupled limit cycle oscillators.
We show that the short-range Kuramoto model can be reduced to a system of ODEs in Def.~\ref{def:quantum_solvable_ODE}, which is one of the promising application candidates for quantum computers, since the treatment of short-range couplings (oscillators embedded in a lattice with nearest-neighbor interactions) presents formidable difficulties both at the analytical and numerical level~\cite{RevModPhys.77.137}.

\subsection{Classical harmonic oscillator}

Let us given the following classical harmonic oscillators; for $j \in [N]$,
\begin{eqnarray}
    m_j \ddot{x}_j = \sum_{k\neq j} \kappa_{jk} (x_k - x_j) - \kappa_{jj}x_j,
\end{eqnarray}
where $m_j > 0, \kappa_{jj} > 0$ and $\kappa_{jk} = \kappa_{kj} > 0$.
Introducing the following variable transformations,
\begin{eqnarray}
    X_j &=& \sqrt{\kappa_{jj}} x_j, \\
    Y_{jk} &=& \sqrt{\kappa_{jk}} (x_j - x_k), (k > j), \\
    V_j &=& \sqrt{m_j} \dot{x}_j,
\end{eqnarray}
we can derive the following equations,
\begin{eqnarray}
    \dot{X}_j &=& \sqrt{\frac{\kappa_{jj}}{m_j}} V_j, \label{C_X} \\
    \dot{Y}_{jk} &=&  \sqrt{\frac{\kappa_{jk}}{m_j}} V_j - \sqrt{\frac{\kappa_{jk}}{m_k}} V_k, \label{C_Y} \\
    \dot{V}_j &=& \sum_{k\neq j} \sqrt{\frac{\kappa_{jk}}{m_j}} Y_{jk} - \sqrt{\frac{\kappa_{jj}}{m_j}} X_j, \label{C_V}
\end{eqnarray}
which trivially satisfy Eq.~(\ref{condition1}). It is shown that this linear system satisfies Eq.~(\ref{condition2}) as follows.
Checking that
\begin{eqnarray}
    \sum_j \sum_{k > j} Y_{jk} \dot{Y}_{jk} 
    &=& \sum_j \sum_{k > j} \kappa_{jk} (x_j - x_k)(\dot{x}_j - \dot{x}_k), \\
    &=& \sum_j \sum_{k \neq j} \kappa_{jk} (x_j - x_k) \dot{x}_j,
\end{eqnarray}
then we can derive that 
\begin{eqnarray}
    \sum_j \left( X_j \dot{X}_j + \sum_{k > j} Y_{jk} \dot{Y}_{jk} + V_j \dot{V}_j \right) = 0,
\end{eqnarray}
which corresponds to the energy conservation.

\subsection{Undamped and unforced Duffing equations}

Let us given the following nonlinear classical harmonic oscillators: for $j \in [N], S_j \subset [N]$,
\begin{eqnarray}
    m_j \ddot{x}_j &=& - \kappa_j x_j - 2\lambda_j x_j^3 \nonumber \\
    &&- \sum_{k \in S_j} \left( \kappa_{jk} (x_j - x_k) + 2\lambda_{jk} (x_j - x_k)^3 \right),
\end{eqnarray}
where $m_j, \kappa_j, \lambda_j > 0$, $\kappa_{jk} = \kappa_{kj} > 0$, and $\lambda_{jk} = \lambda_{kj} > 0$.
Introducing the following variable transformations,
\begin{eqnarray}
    X_j &=& \sqrt{\kappa_j} x_j, \\
    Y_j &=& \sqrt{\lambda_j} x_j^2, \\
    X_{jk} &=& \sqrt{\kappa_{jk}} (x_j - x_k), (k > j), \\
    Y_{jk} &=& \sqrt{\lambda_{jk}} (x_j - x_k)^2, (k > j), \\
    V_j &=& \sqrt{m_j} \dot{x}_j,
\end{eqnarray}
we can derive the following equations,
\begin{widetext}
\begin{eqnarray}
    \dot{X}_j 
    &=& F^{(1)}_j := \sqrt{\frac{\kappa_j}{m_j}} V_j, \\
    \dot{Y}_j 
    &=& F^{(2)}_j := 2\sqrt{\frac{\lambda_j}{m_j \kappa_j}} X_j V_j, \\
    \dot{X}_{jk} 
    &=& F^{(3)}_{jk} := \sqrt{\frac{\kappa_{jk}}{m_j}} V_j - \sqrt{\frac{\kappa_{jk}}{m_k}} V_k, \\
    \dot{Y}_{jk} 
    &=&  F^{(4)}_{jk} := 2\sqrt{\frac{\lambda_{jk}}{m_j \kappa_{jk}}} X_{jk} V_j - 2\sqrt{\frac{\lambda_{jk}}{m_k \kappa_{jk}}} X_{jk} V_k, \\
    \dot{V}_j 
    &=& F^{(5)}_j := - \sqrt{\frac{\kappa_j}{m_j}} X_j - 2\sqrt{\frac{\lambda_j}{m_j \kappa_j}} X_j Y_j - \sum_{k\neq j} \left( \sqrt{\frac{\kappa_{jk}}{m_j}} X_{jk} + 2\sqrt{\frac{\lambda_{jk}}{m_j \kappa_{jk}}} X_{jk}Y_{jk} \right),
\end{eqnarray}
\end{widetext}
from which it is trivial that
\begin{eqnarray}
    \frac{\partial F^{(1)}_j}{\partial X_j} = \frac{\partial F^{(2)}_j} {\partial Y_j} = \frac{\partial F^{(3)}_{jk}}{\partial X_{jk}} = \frac{\partial F^{(4)}_{jk}}{\partial Y_{jk}} = \frac{\partial F^{(5)}_j}{\partial V_j} = 0.
\end{eqnarray}
Furthermore, checking that
\begin{widetext}
\begin{eqnarray}
    \sum_j X_j \dot{X}_j 
    &=& \sum_j \sqrt{\frac{\kappa_j}{m_j}} X_j V_j, \\
    \sum_j Y_j \dot{Y}_j 
    &=& 2 \sum_j \sqrt{\frac{\lambda_j}{m_j \kappa_j}} X_j Y_j V_j, \\
    \sum_{k > j} X_{jk} \dot{X}_{jk} 
    &=& \sum_{k \neq j} \sqrt{\frac{\kappa_{jk}}{m_j}} X_{jk} V_j, \\
    \sum_{k > j} Y_{jk} \dot{Y}_{jk} 
    &=& 2 \sum_{k \neq j} \sqrt{\frac{\lambda_{jk}}{m_j \kappa_{jk}}} X_{jk} Y_{jk} V_j, \\
    \sum_j V_j \dot{V}_j &=& 
    \sum_j \Biggl[ - \sqrt{\frac{\kappa_j}{m_j}} X_j V_j - 2\sqrt{\frac{\lambda_j}{m_j \kappa_j}} X_j Y_j V_j - \sum_{k\neq j} \left( \sqrt{\frac{\kappa_{jk}}{m_j}} X_{jk} V_j + 2\sqrt{\frac{\lambda_{jk}}{m_j \kappa_{jk}}} X_{jk}Y_{jk}V_j \right) \Biggr],
\end{eqnarray}
\end{widetext}
we can derive that the sum of those is zero.

\subsection{Kuramoto model}

The short-range Kuramoto model~\cite{RevModPhys.77.137} is given in the following.

\begin{definition}
    The short-range Kuramoto model has the following governing equations: for $i \in [N], S_i \subset [N]$,
    \begin{eqnarray}
        \dot{\theta}_i = \omega_i - \frac{K}{N} \sum_{j \in S_i} \sin (\theta_i - \theta_j),
    \end{eqnarray}
    where the system is composed of $N$ limit-cycle oscillators, with phases $\theta_i$ and coupling constant $K$.
\end{definition}

We can reduce the short-range Kuramoto model to a system of ODEs in Def.~\ref{def:quantum_solvable_ODE} as follows.
For given variables $x_i, y_i, z_i, w_i, i \in [N]$, 
let us consider the following initial value problem,
\begin{eqnarray}
    \dot{x}_i &=& F_i(y,z,w) \nonumber \\
    &:=& - \omega_i y_i  
    + \frac{K}{N} \sum_{j \in S_i} y_i \left( 
        w_i z_j - z_i w_j
    \right), \label{eq:F} \\
    \dot{y}_i &=& G_i(z,w,x) \nonumber \\
    &:=&   \omega_i x_i 
    - \frac{K}{N} \sum_{j \in S_i} x_i \left( 
        w_i z_j - z_i w_j
    \right), \label{eq:G} \\
    \dot{z}_i &=& H_i(w,x,y) \nonumber \\
    &:=& - \omega_i w_i 
    + \frac{K}{N} \sum_{j \in S_i} w_i \left( 
        y_i x_j - x_i y_j
    \right), \label{eq:H} \\
    \dot{w}_i &=& I_i(x,y,z) \nonumber \\
    &:=& \omega_i z_i 
    - \frac{K}{N} \sum_{j \in S_i} z_i \left( 
        y_i x_j - x_i y_j
    \right), \label{eq:I}
\end{eqnarray}
where the initial conditions are $x_i(0) + z_i(0) = y_i(0) + w_i(0) = 0$ and $x_i(0)^2 + y_i(0)^2 = z_i(0)^2 + w_i(0)^2 = 1$. 
From the above definition, for any $i \in [N]$, it clearly follows that
\begin{eqnarray}
    \frac{\partial F_i}{\partial x_i} + \frac{\partial G_i}{\partial y_i} + \frac{\partial H_i}{\partial z_i} + \frac{\partial I_i}{\partial w_i} = 0, \label{c1} \\
    \frac{d}{dt} (x_i^2 + y_i^2) = \frac{d}{dt} (z_i^2 + w_i^2) = 0. \label{c2}
\end{eqnarray}
Furthermore, noting that $w_i(0)z_j(0) = y_i(0)x_j(0)$ under the initial condition $x_i(0) + z_i(0) = y_i(0) + w_i(0) = 0$, we can derive that
\begin{eqnarray}
    \frac{d}{dt} (x_i + z_i) = \frac{d}{dt} (y_i + w_i) = 0. \label{c3}
\end{eqnarray}
From Eq.~(\ref{c2}) and Eq.~(\ref{c3}), the initial conditions are maintained throughout the dynamics.

Let us consider the following variable transformations,
\begin{eqnarray}
    x_i &=& \cos \theta_i, \\
    y_i &=& \sin \theta_i, \\
    z_i &=& \cos (\theta_i + \pi), \\
    w_i &=& \sin (\theta_i + \pi),
\end{eqnarray}
which satisfy that for any $\theta_i$,
\begin{eqnarray}
    x_i + z_i = y_i + w_i = 0, \\
    x_i^2 + y_i^2 = z_i^2 + w_i^2 = 1.
\end{eqnarray}
Substituting the variable transformations into Eq.~(\ref{eq:F}), Eq.~(\ref{eq:G}), Eq.~(\ref{eq:H}), and Eq.~(\ref{eq:I}), we obtain the following two equations,
\begin{eqnarray}
    F_i(y,z,w) &=& -H_i(w,x,y) \nonumber \\
    &=& -\sin \theta_i \left( \omega_i -\frac{K}{N} \sum_{j \in S_i} \sin (\theta_i - \theta_j) \right), \nonumber \\ \\
    G_i(z,w,x) &=& -I_i(x,y,z) \nonumber \\
    &=&  \cos \theta_i \left( \omega_i -\frac{K}{N} \sum_{j \in S_i} \sin (\theta_i - \theta_j) \right),
\end{eqnarray}
Noting that 
\begin{eqnarray}
    \frac{d}{dt} \cos \theta_i = -\frac{d}{dt} \cos (\theta_i + \pi) &=& -\sin \theta_i \frac{d\theta_i}{dt}, \\
    \frac{d}{dt} \sin \theta_i = -\frac{d}{dt} \sin (\theta_i + \pi) &=& \cos \theta_i \frac{d\theta_i}{dt},
\end{eqnarray}
we see that variables $\theta_i$ follow the Kuramoto model.

\section{Conclusion} \label{Sec:7}

In this study, we present a class of nonlinear ordinary differential equations (ODEs) solvable with the Koopman-von Neumann embedding in the form of mapping to Hamiltonian dynamics.
And we showed that the proposed ODEs can be efficiently solvable with a computational complexity of $T {\rm log}(N) {\rm polylog}(1/\epsilon)$, where $T$ is the evolution time, $\epsilon$ is the allowed error, and $N$ is the number of variables in the proposed ODEs. 

The major challenges were that the mapped Hamiltonian is not sparse in general and the max norm and spectral norm of the mapped quantum state are not always conserved. Furthermore, the nonlinearity of the proposed ODEs we can handle was forced to be assumed to be sufficiently small, since which implicitly depends on the max norm of the mapped Hamiltonian.

To overcome the problem, we found a concrete form of nonlinear ODEs for which the mapped Hamiltonian is sparse and the norm of the quantum state is preserved.
And as the result, we proved that if $m$ is sufficiently larger than $b$, an approximated position state $\ket{x}$ in the finite $m$ occupation number representation is sufficient to obtain a degree $b$ polynomial output. 
Furthermore, we expanded the Hamiltonian dynamical system in an orthogonal polynomial basis and approximated it in a truncated subspace to save the quantum memory cost.
Compared to the quantum Carleman linearization~\cite{Liu2021-wq}, the proposed ODEs are embedded into Hamiltonian dynamics and applied to the Hamiltonian simulation~\cite{10.1145/3313276.3316366}, thus they depends on neither condition number nor the decay of the solution.

As an application, we showed that the proposed ODEs include classical harmonic oscillators, undamped and unforced Duffing equations, and the short-range Kuramoto model. 
Recently, a quantum algorithm has been proposed for solving coupled harmonic oscillators, a classical mechanical system, with exponential quantum speedup~\cite{Babbush2023-ju}. 
While their study can only handle essentially linear systems, our research is capable of addressing nonlinear systems, enabling exponential quantum speedup with respect to the number of variables $N$ in the system. 
Furthermore, the previous study assumes an oracle for parameters such as spring constants
for Hamiltonian simulation, while our research successfully presents conditions for mapping to sparse Hamiltonian dynamics from specific partial differential equation constraints. Actually, the coupled harmonic oscillators can be reduced into the proposed ODEs, due to the linearity of the system there is no error from the truncation. 

Although our method has been developed with applications to nonlinear dynamics in mind, its potential use for solving large linear systems represents an intriguing avenue for future research.
Our result finding conditions under which large scale nonlinear differential equations can be solved eﬀiciently using a quantum computer makes an important contribution to the application of quantum computers to nonlinear problems.

\begin{acknowledgments}
We would like to thank Kaoru Mizuta for his valuable comment on evaluating the truncation error.
K.F. is supported by MEXT Quantum Leap Flagship Program (MEXT Q-LEAP) Grant No. JPMXS0118067394 and JPMXS0120319794, JST COI-NEXT Grant No. JPMJPF2014.
\end{acknowledgments}

\appendix


\section{Proof of Lemma~\ref{lem:1}} \label{apdx:A}

For convenience, denote the local operator given by Eq.~(\ref{eq:q_sol_H}) as 
\begin{eqnarray}
    \hat{f}_p = \sum_{i \in p} \alpha_{p\to i} \hat{k}_i \prod_{j \in p\setminus i} \hat{x}_j,
\end{eqnarray}
and $\mathcal{H}_m$ as the Hilbert space restricted to the subspace bounded by the total occupation number $m$.

First, we prove the sparsity of the Hamiltonian $\hat{H}_m$.
Note that, from $\hat{k} \ket{0} \propto \hat{x}\ket{0}$ and the definition of $\alpha_{p\to i}$, it satisfies that 
\begin{eqnarray}
    \sum_{i \in p} \alpha_{p\to i} \hat{k}_i \ket{0} \prod_{j \in p\setminus i} \hat{x}_j \ket{0} = 0. \label{eq:prop1}
\end{eqnarray}

Considering the Hilbert space $\mathcal{H}_m$ restricted to the subspace bounded by the total occupation number $m$, the number of non zero number states is at most $m$.
From Eq.~(\ref{eq:prop1}), we do not need to consider the local Hamiltonian 
\begin{eqnarray}
    \sum_{i \in p} \alpha_{p\to i} \hat{k}_i \prod_{j \in p\setminus i} \hat{x}_j,
\end{eqnarray}
acting on the local vacuum states.
From the definition of $\mathcal{S}$, since the number of $p \in \mathcal{S}$ having the index $i$ is at most $c$ and the size of $p$ is at most $d$ and $\hat{x}$ and $\hat{k}$ are $2$-sparse, then the operator given by
\begin{eqnarray}
    \sum_{\substack{p \in \mathcal{S} \\ : i \in p}} \sum_{j \in p} \alpha_{p\to j} \hat{k}_j \prod_{k \in p\setminus j} \hat{x}_k \label{H_H}
\end{eqnarray}
acts on at most $c2^d$ number states.
Therefore, since the number of $p \in \mathcal{S}$ having the indices of not vacuum states is upper bounded by $m$, the Hamiltonian $\hat{H}_m$ is $O(c2^d m)$-sparse.

Second, we show the max norm of the Hamiltonian $\hat{H}_m$.
Let us consider $p \in \mathcal{S}$ and two number states $\ket{\mathbf{n}}$ and $\ket{\mathbf{n}'}$ such that $\bra{\mathbf{n}} \hat{f}_p \ket{\mathbf{n}'} \neq 0$. 
Remembering that $\hat{x} = (\hat{a} + \hat{a}^{\dagger})/\sqrt{2}$, $\hat{k} = i (\hat{a}^{\dagger} - \hat{a})/\sqrt{2}$ and 
\begin{eqnarray}
    \bra{\mathbf{n}} \hat{f}_p \ket{\mathbf{n}'} 
    &=&
    \sum_{i \in p} \alpha_{p\to i} \bra{n_i} \hat{k}_i \ket{n'_i} \nonumber \\
    && \times \prod_{j \in p \setminus i} \bra{n_j} \hat{x}_j \ket{n'_j} 
    \prod_{k \not\in p} \braket{n_k | n'_k} \\
    &\neq& 0, 
\end{eqnarray}
we derive that $n_k = n'_k$ must be satisfied for $k \not\in p$. 
Thus, if $\bra{\mathbf{n}} \hat{f}_p \ket{\mathbf{n}'} \neq 0$ for any $q \neq p$, then we obtain that
\begin{eqnarray}
    \bra{\mathbf{n}} \hat{f}_q \ket{\mathbf{n}'} = 0. \label{eq:relation}   
\end{eqnarray}
Then, the max norm is upper bounded by
\begin{eqnarray}
    \max_{\mathbf{n},\mathbf{n}'}|\bra{\mathbf{n}} \hat{H}_m \ket{\mathbf{n}'}|
    &=& \max_{\mathbf{n},\mathbf{n}', p \in \mathcal{S}} \left| \bra{\mathbf{n}} \hat{f}_{p} \ket{\mathbf{n}'} \right| \\
    &\le& 
    \max_{p \in \mathcal{S}} \sum_{i \in p} \left| \alpha_{p\to i} \right| \times \prod_{j \in p} \sqrt{\frac{m_j}{2}} \\
    &\le&
    \eta d 
    \left(
        \frac{m}{2}
    \right)^{\frac{d}{2}},
\end{eqnarray}
where $m_j = \max \{n_j, n'_j \}$ and $|n_j - n'_j| = 1$.

Third, we show the spectral norm of the Hamiltonian $\hat{H}_m$.
Let us consider a fixed $(l \le m)$-cite-excited state $\ket{\mathbf{n}}$ and denote $q_l \subset [N]$ as the index subset of the $l$ cites. 
For the state $\ket{\mathbf{n}}$ such that $n_i \neq 0, (i \in q_l)$ and $n_j = 0, (j \not\in q_l)$, we can derive that
\begin{eqnarray}
    \max_{\mathbf{n}'} \left| \bra{\mathbf{n}'} \hat{f}_{p} \ket{\mathbf{n}} \right| 
    &\le&
    \sum_{i \in p} \left| \alpha_{p\to i} \right| \times \prod_{j \in p} \sqrt{\frac{n_j + 1}{2}} \\
    &\le&
    \eta d \prod_{j \in p} \sqrt{\frac{n_j + 1}{2}} \\ 
    &\le&
    \eta d \left[ \sum_{j \in p} \frac{n_j+1}{2|p|} \right]^{\frac{|p|}{2}} \\
    &=&
    \eta d \left[ \frac{\bar{n}_{p,\mathbf{n}}+1}{2} \right]^{\frac{|p|}{2}},
\end{eqnarray}
where $\bar{n}_{p,\mathbf{n}} = \sum_{i \in p} n_i / |p|$.

From the above and Eq.~(\ref{eq:prop1}), for the $(l \le m)$-cite-excited state $\ket{\mathbf{n}}$, we derive that
\begin{eqnarray}
    \sum_{\mathbf{n'}} | \bra{\mathbf{n}'} \hat{H}_m \ket{\mathbf{n}} | 
    &=&
    \sum_{\mathbf{n}', p} | \bra{\mathbf{n}'} \hat{f}_p \ket{\mathbf{n}} | \label{eq:line_2} \\
    &=&
    \sum_{\mathbf{n'}} \sum_{p:p \cap q_l \neq \emptyset} |\bra{\mathbf{n}'} \hat{f}_p \ket{\mathbf{n}}| \\
    &\le& 
    \eta d \sum_{p:p \cap q_l \neq \emptyset} \left[ \frac{\bar{n}_{p,\mathbf{n}}+1}{2} \right]^{\frac{|p|}{2}} \\
    &=& 
    \eta d \sum_{p:p \cap q_l \neq \emptyset} m^{\frac{|p|}{2}} \left[ \frac{\bar{n}_{p,\mathbf{n}}+1}{2m} \right]^{\frac{|p|}{2}} \nonumber \\ \\
    &\le& 
    \eta d m^{\frac{d}{2}} \sum_{p:p \cap q_l \neq \emptyset} \frac{\bar{n}_{p,\mathbf{n}}+1}{2m} \\
    &\le& 
    \eta cd m^{\frac{d}{2}}, \label{eq:last_line}
\end{eqnarray}
where we used Eq.~(\ref{eq:relation}) to derive Eq.~(\ref{eq:line_2}) and we used that the number of index subsets containing an index $i$ is at most $c$ to derive Eq.~(\ref{eq:last_line}).
Therefore,
\begin{eqnarray}
    || \hat{H}_m ||_1 
    = \max_{\mathbf{n}} \sum_{\mathbf{n}'} | \bra{\mathbf{n}'} H \ket{\mathbf{n}} | \le \eta cd m^{\frac{d}{2}}. \label{eq:1_norm_H}
\end{eqnarray}
Since $\hat{H}_m$ is a hermitian, $|| \hat{H}_m ||_{\infty}$ is also upper bounded by Eq.~(\ref{eq:1_norm_H}).
From the relation $|| H ||_2 \le \sqrt{|| H ||_1 || H ||_{\infty}}$, the claim follows.


\section{Proof of Theorem \ref{thm:main}} \label{apdx:B}

In preparation for proving Theorem~\ref{thm:main}, we prove the following two lemmas.

\begin{lemma} \label{lem:6}
    Let us consider a system of $N$-dimensional ODEs in Def.~\ref{def:quantum_solvable_ODE} and the mapped Hamiltonian
    \begin{eqnarray}
        \hat{H} &=& \hat{H}^{(0)} + \hat{H}^{(I)} \\
        &=& \sum_{\substack{i \in p, p \in \mathcal{S} \\ : |p|=2}} \alpha _{p \to i} \hat{k}_i \prod_{j \in p\setminus i} \hat{x}_j \nonumber \\
        && + \sum_{\substack{i \in p, p \in \mathcal{S} \\ : |p| > 2}} \alpha _{p \to i} \hat{k}_i \prod_{j \in p\setminus i} \hat{x}_j, i \in [N].
    \end{eqnarray}
    Then, for arbitrary quantum states $\ket{\phi_l} \in \mathcal{H}_l$ and $\ket{\phi_{l'}} \in \mathcal{H}_{l'}$ respectively, 
    \begin{eqnarray}
        | \bra{\phi_l} \hat{H}(t_n) \cdots \hat{H}(t_1) \ket{\phi_{l'}} | =
        \begin{cases}
            0 & (n < n_0), \\
            \le \lVert \hat{H}^{(I)} \rVert_2^n & (n \ge n_0),
        \end{cases} \nonumber \\ 
    \end{eqnarray}
    where $n_0 = \lceil |l-l'| / d \rceil$ and $\hat{H}(t) = e^{i \hat{H}^{(0)} t} \hat{H}^{(I)} e^{-i \hat{H}^{(0)} t}$.
\end{lemma}

\begin{proof}
    Note that $\hat{H}^{(0)}$ conserves the total occupation number from the definition of the proposed ODE in Def.~\ref{def:quantum_solvable_ODE}, which is checked by the fact that $\hat{H}^{(0)}$ consists of the sum of the operators
    $
        \alpha_{p\to i} (\hat{k}_i \otimes \hat{x}_j - \hat{x}_i \otimes \hat{k}_j) = i \alpha_{p\to i} (\hat{a}_i^{\dagger} \otimes \hat{a}_j - \hat{a}_i \otimes \hat{a}_j^{\dagger}).
    $
    Furthermore, since $\hat{H}^{(I)}$ consists of the sum of products of at most $d$ creation or annihilation operators, $| \bra{\phi_l} (\hat{H}^{(I)})^n \ket{\phi_{l'}} | = 0$ when $n < \lceil |l-l'| / d \rceil$.
    
    Then, let us show the claim by induction. For $n=1$, since $| \bra{\phi_l} \hat{H}(t_1) \ket{\phi_{l'}} | = | \bra{\phi_l(t_1)} \hat{H}^{(I)} \ket{\phi_{l'}(t_1)} |$, the claim clearly follows. 
    Assume that the claim is satisfied for $n=k$. 
    Considering the following state
    \begin{widetext}
    \begin{eqnarray}
        \hat{H}(t_{k+1}) \ket{\phi_l} &=& \sqrt{|\bra{\phi_l(t_{k+1})} (\hat{H}^{(I)})^2 \ket{\phi_l(t_{k+1})}|} \frac{\hat{H}(t_{k+1}) \ket{\phi_l}}{\sqrt{|\bra{\phi_l(t_{k+1})} (\hat{H}^{(I)})^2 \ket{\phi_l(t_{k+1})}|}} \\
        &=:& \sqrt{|\bra{\phi_l(t_{k+1})} (\hat{H}^{(I)})^2 \ket{\phi_l(t_{k+1})}|} \ket{\phi_{l-d}^{l+d}(t_{k+1})},
    \end{eqnarray}
    \end{widetext}
    $\ket{\phi_{l-d}^{l+d}(t_{k+1})}$ is a quantum state in the Hilbert space restricted to the subspace of the total occupation numbers from $l-d$ to $l+d$. Therefore, since $\sqrt{|\bra{\phi_l(t_{k+1})} (\hat{H}^{(I)})^2 \ket{\phi_l(t_{k+1})}|} \le \| \hat{H}^{(I)} \|_2$, the claim follows.
    
\end{proof}

\begin{lemma} \label{lem:7}
    Let us consider a system of $N$-dimensional ODEs in the interval $[0,T]$ in Def.~\ref{def:quantum_solvable_ODE} and define the rescaled mapped Hamiltonian
    \begin{eqnarray}
        \hat{H} &=& \hat{H}^{(0)} + \hat{H}^{(I)} \\
        &=& \sum_{\substack{i \in p, p \in \mathcal{S} \\ : |p|=2}} \alpha _{p \to i} \hat{k}_i \prod_{j \in p\setminus i} \hat{x}_j \nonumber \\
        && + 
        \sum_{\substack{i \in p, p \in \mathcal{S} \\ : |p| > 2}} \frac{\alpha _{p \to i}}{\Delta^{|p|-2}} \hat{k}_i \prod_{j \in p\setminus i} \hat{x}_j, i \in [N],
    \end{eqnarray}
    where $\Delta = C T m^d$, $C = 2\eta^2 cd^3 2^{-d/2}$, and $\eta = \max_{i \in p, p \in \mathcal{S}} |\alpha_{p\to i}|$.
    Then, for the Hamiltonian $\hat{H}_m = \hat{H}^{(0)}_m + \hat{H}^{(I)}_m$ on $\mathcal{H}_{m}$, arbitrary states $\ket{\phi_l} \in \mathcal{H}_l$ and $\ket{\phi_{l'}} \in \mathcal{H}_{l'}$ respectively, for $t \in [0, \lVert \hat{H}_m \rVert_{\max} T]$,
    \begin{eqnarray}
        | \bra{\phi_l} e^{i \hat{H}_m t} \ket{\phi_{l'}} | \le \frac{2}{n_0!} \left( \frac{1}{2d} \right)^{n_0},
    \end{eqnarray}
    where $n_0 = \lceil |l-l'| / d \rceil$.
\end{lemma}

\begin{proof}
    We follow the proof technique in Ref.~\cite{Mizuta2023optimalhamiltonian}.
    We first employ the interaction picture. 
    Defining the unitary operator
    \begin{eqnarray}
        \mathcal{U}_0(t) &=& e^{-i \hat{H}^{(0)}_m t}, 
    \end{eqnarray}
    the time evolution operator $e^{-i \hat{H}_m t}$ is represented as $e^{-i \hat{H}_m t} = \mathcal{U}_0(t) \ \mathcal{U}_I(t)$, where 
    \begin{eqnarray}
        \mathcal{U}_I(t) &=& \mathcal{T}\exp \left( -i \int_0^t dt' \hat{H}_I(t') \right), \\
        \hat{H}_I(t) &=& \mathcal{U}^{\dagger}_0(t) \hat{H}^{(I)}_m \mathcal{U}_0(t).
    \end{eqnarray}
    Then, using the Dyson series expansion, we derive that
    \begin{widetext}
    \begin{eqnarray}
        | \bra{\phi_l} e^{i \hat{H}_m t} \ket{\phi_{l'}} | &=& | \bra{\phi_l} \mathcal{U}_0(t) \mathcal{U}_I(t)  \ket{\phi_{l'}} | = | \bra{\phi_l(t)} \mathcal{U}_I(t)  \ket{\phi_{l'}} | \\
        &\le& \sum_{n=0}^{\infty} \int_0^t dt_n \cdots \int_0^{t_2} dt_1 | \bra{\phi_l(t)} \hat{H}_I(t_n) \cdots \hat{H}_I(t_1) \ket{\phi_{l'}}| \\
        &=& \sum_{n=n_0}^{\infty} \int_0^t dt_n \cdots \int_0^{t_2} dt_1 | \bra{\phi_l(t)} \hat{H}_I(t_n) \cdots \hat{H}_I(t_1) \ket{\phi_{l'}}| \label{eq:A_1} \\
        &\le& 
        \sum_{n=n_0}^{\infty} \frac{t^n}{n!} \| \hat{H}^{(I)}_m \|^n_2 \label{eq:A_2} \\
        &\le& \sum_{n=n_0}^{\infty} \frac{1}{n!} \left( \gamma cd m^{d/2} t \right)^n \label{eq:A_3} \\
        &\le& \frac{2}{n_0!} \left( \frac{1}{2d} \right)^{n_0},
    \end{eqnarray}
    \end{widetext}
    where $\gamma = \max_{i \in p, p \in \mathcal{S}:|p|>2} |\alpha_{p\to i} / \Delta^{|p|-2}|$.
    To derive Eq.~(\ref{eq:A_1}) and Eq.~(\ref{eq:A_2}), we used the result of Lemma~\ref{lem:6}. To derive Eq.~(\ref{eq:A_3}), we used the result of Lemma~\ref{lem:1}. To derive the last inequality, we used the following inequality \cite{Mizuta2023optimalhamiltonian};
    \begin{eqnarray}
        \sum_{n=n_0}^{\infty} \frac{x^n}{n!} \le 2 \frac{x^{n_0}}{n_0!}, \ \mathrm{if} \ n_0 \ge 2x \ge 0,
    \end{eqnarray}
    and the following relation
    \begin{eqnarray}
        2 \gamma cdm^{\frac{d}{2}}t &<& \frac{2 \eta cd m^{\frac{d}{2}} \lVert \hat{H}_m \rVert_{\max}T}{\Delta} \\
        &<& \frac{2 \eta^2 c d^2 2^{-d/2} T m^d }{\Delta} \\
        &=& \frac{1}{d} \le \frac{|l-l'|}{d}.
    \end{eqnarray}
\end{proof}

    Then, we prove Theorem~\ref{thm:main}.
    Basically, we follow the proof given in Ref.~\cite{Engel2021-vl}.
    For the original system of ODEs in Def.~\ref{def:quantum_solvable_ODE}, let us consider the rescaled system of ODEs and divide it into a linear part and a non-linear part as
    \begin{eqnarray}
        \frac{dx_i}{dt} = F^{(0)}_i(\textbf{x}) + F^{(I)}_i(\textbf{x}), i \in [N], \label{eq:thm1_ode}
    \end{eqnarray}
    where 
    \begin{eqnarray}
        F^{(0)}_i(\textbf{x}) 
        &=&
        \sum_{\substack{p \in \mathcal{S} \\ : i \in p, |p|=2}} 
        \alpha _{p \to i} \prod_{j \in p\setminus i} x_j, \label{eq:F1} \\
        F^{(I)}_i(\textbf{x}) 
        &=& 
        \sum_{\substack{p \in \mathcal{S} \\ : i \in p, |p|>2}} 
        \frac{\alpha _{p \to i}}{\Delta^{|p|-2}} \prod_{j \in p\setminus i} x_j,
    \end{eqnarray}
    and define that
    \begin{eqnarray}
        \gamma
        = 
        \max_{\substack{i \in p, p \in \mathcal{S} \\ : |p|>2}} 
        \left|\frac{\alpha _{p \to i}}{\Delta^{|p|-2}} \right|.
    \end{eqnarray}
    Similarly, we can derive the representation of the Hamiltonian and we denote that $\hat{H} = \hat{H}^{(1)} + \hat{H}^{(I)}$.

    First, for a quantum state $\ket{c} \in \mathcal{H}_{b}$, from Lemma~\ref{lem:7}, 
    \begin{eqnarray}
        |\braket{c | \psi_m(t)}| &\le& \left| \sum_{n=0}^{n_0-1} \frac{\bra{c}(-i \hat{H}_m t)^n \ket{\psi_m(0)}}{n!} \right| \nonumber \\
        && + \frac{2}{n_0!} \left( \frac{1}{2d} \right)^{n_0} \\
        &=:&  
        \left| \sum_{n=0}^{n_0-1} c_{m,n} \gamma^n \right| + \frac{2}{n_0!} \left( \frac{1}{2d} \right)^{n_0}, \label{eq:thm1_1}
    \end{eqnarray}
    where $n_0 = \lceil (m - b)/d \rceil$. Thus, the second term of Eq.~(\ref{eq:thm1_1}) can be ignored if a sufficiently large $m$ is taken.
    For instance, we find a sufficiently large $m$ to satisfy the following inequality;
    \begin{eqnarray}
        \frac{2}{n_0!} \left( \frac{1}{2d} \right)^{n_0} < \frac{\epsilon}{\Delta^b} \label{ineq:1}
    \end{eqnarray}
    where remember that $\Delta = CTm^d$, $C = 2\eta^2 c d^3 2^{-d/2}$, and $m = dn_0 + b$. Using the Stirling's formula, we derive that
    \begin{eqnarray}
        \left(n_0 + \frac{1}{2} \right) \log n_0 - n_0 + n_0 \log 2d + \log \sqrt{\pi/2} \nonumber \\
        > \log \epsilon^{-1} + b \log (CT) + bd \log (n_0(d+1)) \nonumber \\ \\
        \Rightarrow \left(n_0 - bd + \frac{1}{2} \right) \log n_0 + n_0 (\log 2d - 1) \nonumber \\
        > \log \epsilon^{-1} + b \log (CT) + \log \left( \frac{(d+1)^{bd}}{\sqrt{\pi/2}} \right). \nonumber \\
    \end{eqnarray}
    Then, we can derive one sufficient condition to satisfy Eq.~(\ref{ineq:1}),
    \begin{eqnarray}
        n_0 &=& \max \Biggl\{ \log \epsilon^{-1}, b\log (CT), \nonumber \\
        && \left. \log \left( \frac{(d+1)^{bd}}{\sqrt{\pi/2}} \right), 4bd - 2, e^4 \right\}.
    \end{eqnarray}

    Second, let us examine the exact output quantity.
    We can expand $\braket{c | \psi(t,\gamma)}$ as a power series in $\gamma$:
    \begin{eqnarray}
        \braket{c | \psi(t,\gamma)} &=& \sum_{n=0}^{\infty} \frac{\bra{c}(-i\hat{H}t)^n\ket{\psi(0)}}{n!} \\
        &=& \sum_{n=0}^{\infty} c_n(t) \gamma^n. \label{eq:thm1_2}
    \end{eqnarray}
    To show the convergence of Eq.~(\ref{eq:thm1_2}), consider an initial value problem given by Eq.~(\ref{eq:thm1_ode}) and ${\bf x}(0,\gamma) = {\bf x}_0$.
    Since Eq.~(\ref{eq:thm1_ode}) is infinitely differentiable in $\gamma$, the $\gamma = 0$ solution is ${\bf x}(t,0) = e^{{\bf F}^{(0)}t}{\bf x}_0$, and we consider the finite simulation time, then from the regular perturbation theorem~\cite{Hoppensteadt_undated-cc} the solution ${\bf x}(t,\gamma)$ can be expanded in Taylor expansion,
    \begin{eqnarray}
        {\bf x}(t,\gamma) = \sum_{n=0}^{n_0 - 1} {\bf x}_n(t)\gamma^n + O(\gamma^{n_0}) \ \mathrm{as} \ \gamma \to 0, \label{eq:thm1_3}
    \end{eqnarray}
    and $\gamma^{n_0}$ can be negligible if $m$ is sufficiently large.
    For instance, we find a sufficiently large $m$ to satisfy the following inequality;
    \begin{eqnarray}
        \gamma^{n_0} &<& \frac{\epsilon}{\Delta^b} \label{ineq:2} \\
        \Rightarrow \left( \frac{\eta}{\Delta} \right)^{n_0} &<& \frac{\epsilon}{\Delta^b} \\
        \Rightarrow n_0 \log (\Delta / \eta) &>& \log \epsilon^{-1} + b \log \Delta \\
        \Rightarrow (n_0 - b) \log (\Delta / \eta) &>& \log \epsilon^{-1} + b \log \eta.
    \end{eqnarray}
    Noting that
    \begin{eqnarray}
        \frac{\Delta}{\eta} &>& e^3 \\
        \Rightarrow n_0 &>& \frac{1}{d} \left( \frac{\eta e^3}{C} \right)^{\frac{1}{d}} - \frac{b}{d} \\ 
        &=& \frac{\sqrt{2}}{d} \left( \frac{e^3}{2\eta c d^3 T} \right)^{\frac{1}{d}} - \frac{b}{d},
    \end{eqnarray}
    we can derive one sufficient condition to satisfy Eq.~(\ref{ineq:2}),
    \begin{eqnarray}
        n_0 &&= \max \Biggl\{ \log \epsilon^{-1}, b \log \eta, 3b, \nonumber \\
        && \left. \frac{\sqrt{2}}{d} \left( \frac{e^3}{2\eta c d^3 T} \right)^{\frac{1}{d}} - \frac{b}{d} \right\}.
    \end{eqnarray}
    
    Remembering that $\braket{c | \psi(t,\gamma)}$ is a specified polynomial of ${\bf x}(t,\gamma)$ in the occupation number representation, we can apply Eq.~(\ref{eq:thm1_3}) to this polynomial to obtain
    \begin{eqnarray}
        \braket{c | \psi(t,\gamma)} = \sum_{n=0}^{n_0-1} c_n(t) \gamma^n + O(\gamma^{n_0}).
    \end{eqnarray}

    Finally, we show that $c_{m,n}(t) = c_n(t)$ for $n < n_0$ if $\hat{H}$ is derived from the ODEs consisting of degree $d-1$ polynomials and $m \ge d n_0 + b$ is satisfied.
    Since it takes $n_0 +1$ applications of $\hat{H}$ to couple from $\ket{c}$ to any component with $n > m$, all terms in
    \begin{eqnarray}
        \braket{c | \psi(t,\gamma)} = \bra{c} e^{-i\hat{H}t} \ket{\psi(0)}
    \end{eqnarray}
    that are affected by the truncation have at least $n_0 +1$ factors of $\gamma$, which implies the result.

\section{Initial state preparation from another access model}

The initial state preparation from a more basic access model is an interesting problem.
Here is a simple attempt. 
For $M = {}_{N+m}C_m$, let us assume an oracle that prepares a superposition state 
\begin{eqnarray}
    \ket{\phi} = \sum_{i = 0}^{M-1} c_i \ket{i}, \label{eq:index}
\end{eqnarray}
and consider the problem of converting the state into 
\begin{eqnarray}
    \ket{\phi_m} = \sum_{i = 0}^{M-1} c_i \ket{\mathbf{n}_i}, \label{eq:superpose_ex}
\end{eqnarray}
where $\mathbf{n}_i$ is the $i$-th $m$ integers taken from $N+1$ integers $\{ 0,1\cdots, N \}$, allowing for duplicates, and arranged in ascending order.
Then, there exists an efficient algorithm to exchange the indices $i$ and the ascending ordered combinations $\mathbf{n}_i$.

\begin{lemma} \label{lemma:i2comb}
    Let $M={}_{N+m}C_m$ be the total number of combinations that take $m$ integers, allowing for duplicates from $N+1$ integers $\{ 0,1,\cdots, N \}$, and the combination 
    \begin{eqnarray}
        \mathbf{m} = \underbrace{\overbrace{i_1 \cdots i_1}^{n_{i_1}}\overbrace{i_2 \cdots i_2}^{n_{i_2}} \cdots \overbrace{i_l \cdots i_l}^{n_{i_l}}}_{m}
    \end{eqnarray}
    of the integers be defined to be arranged in ascending order. For a given index $i \in \{ 0,\cdots, M-1 \}$, there exists an algorithm $f$ calculating the combination of integers corresponding to $i$ in $O(poly(m\log N))$ time complexity. The inverse algorithm is in the same way.
\end{lemma}

\begin{proof}
    The total number of combinations that take $m$ integers, allowing for duplicates, from $N+1$ integers $\{ 0,1,\cdots, N \}$ is $M = {}_{N+m}C_m$. 
    Suppose that the left-most integer is $a$, then the number of combinations of the remaining integers is ${}_{N+m-a-1}C_{m-1}$.
    Further, when the left-most integer is greater than or equal to $a$, the number of combinations of the remaining integers is
    \begin{eqnarray}
        {}_{N+m-a-1}C_{m-1} + \cdots + {}_{m-1}C_{m-1} = {}_{N+m-a}C_m.
    \end{eqnarray}
    Therefore, for a given index $i \in \{ 0,1,\cdots, M-1 \}$, the left-most integer $a$ can be determined by using binary search to find the $a$ such that ${}_{N+m-a-1}C_m \le i < {}_{N+m-a}C_m$, where denote that ${}_{m-1}C_m = 0$.
    The same method is then used for index $i - {}_{N+m-a-1}C_m$ to determine the left-second-most integer, which can be repeated to determine the other integers sequentially too.

    Conversely, for a given combination of integers, first, if the left-most integer $a$, then the index $i$ is at least larger than or equal to ${}_{N+m-a-1}C_m$. Second, if the left-second-most integer $b$, then the index $i$ is at least larger than ${}_{N+m-b-2}C_{m-1} + {}_{N+m-a-1}C_m$, which can be repeated to determine the index $i$.
\end{proof}

From Lemma \ref{lemma:i2comb}, under the assumption of quantum arithmetic, there exist two oracles such that  
\begin{eqnarray}
    O_f \ket{i} \otimes \ket{\mathbf{0}} &=& \ket{i} \otimes \ket{\mathbf{m}_i}, \label{oracle:f} \\
    O_{f^{-1}} \ket{i} \otimes \ket{\mathbf{m}_j} &=& \ket{i \oplus j} \otimes \ket{\mathbf{m}_j}. \label{oracle:inv_f}
\end{eqnarray}
From the above consideration, let us show the existence of algorithm transforming $\ket{\phi}$ to $\ket{\phi_m}$.

\begin{lemma}
    For a given state of Eq.~(\ref{eq:index}) through an oracle, under the assumption of quantum arithmetic,
    there exists an algorithm preparing $\ket{\phi_m}$ in Eq.~(\ref{eq:superpose_ex}) with the time complexity $O({\rm poly}(m\log N))$.
\end{lemma}

\begin{proof}
    It is trivial from two oracles Eq.~(\ref{oracle:f}) and Eq.~(\ref{oracle:inv_f}), \textit{i.e.},
    \begin{eqnarray}
        O_{f^{-1}}O_f \left( \sum_{i = 0}^{M-1}  c_i \ket{i} \otimes \ket{\mathbf{0}} \right) = \ket{0} \otimes \sum_{i = 0}^{M-1}  c_i \ket{\mathbf{m}_i}.
    \end{eqnarray}
\end{proof}

\bibliography{qsnde_prr}

\end{document}